\documentclass[11pt,a4paper]{article}

\usepackage[margin=2.5cm]{geometry}

\usepackage{lmodern}
\usepackage{iftex}
\ifpdftex
\usepackage[utf8]{inputenc}
\usepackage[T1]{fontenc}
\usepackage{textcomp}
\fi

\usepackage{amsmath}        
\usepackage{mathtools}
\usepackage{amsfonts}       
\usepackage{amssymb}
\usepackage{amsthm}         
\usepackage{bm}             
\usepackage{booktabs}       
\usepackage{caption}        
\usepackage{dcolumn}        
\usepackage{floatrow}       
\usepackage{graphicx}       
\usepackage[nopatch=item]{microtype}   
\usepackage{paralist}       
\usepackage{enumitem}
\usepackage[nottoc]{tocbibind} 
\usepackage[dvipsnames]{xcolor}         
\usepackage{adjustbox} 
\usepackage{placeins} 
\usepackage{array} 
\usepackage{tabularx} 

\usepackage{hyperref}
\hypersetup{unicode}
\hypersetup{breaklinks=true}

\usepackage{algorithm}
\usepackage[noEnd=true,spaceRequire=false]{algpseudocodex}
\algrenewcommand\algorithmicthen{\hspace{-1mm}\textbf{:}}
\algrenewcommand\algorithmicelse{\textbf{else:}}
\algrenewcommand\algorithmicdo{\hspace{-1mm}\textbf{:}}
\usepackage{fancyvrb}       
\usepackage{thm-restate}	

 \usepackage{cleveref}

\usepackage[natbib,style=iso-alphabetic]{biblatex} 
\DeclareFieldFormat[article]{citetitle}{\mkbibemph{#1}}
\DeclareFieldFormat[article]{title}{\mkbibemph{#1}}
\DeclareFieldFormat[article]{journaltitle}{\normalfont{#1}}

\usepackage{authblk}
\usepackage{orcidlink}

\theoremstyle{plain}
\newtheorem{thm}{Theorem}
\newtheorem{lemma}{Lemma}

\newtheorem{defn}{Definition}

\newtheorem{obs}{Observation}
\newtheorem{fact}{Fact}

\theoremstyle{remark}
\newtheorem*{rem}{Remark}

\ifcsname DeclareCaptionStyle\endcsname
\DeclareCaptionStyle{thesis}{style=base,font=small,labelfont=bf,labelsep=quad}
\fi

\DefineVerbatimEnvironment{code}{Verbatim}{fontsize=\small, frame=single}

\ifcsname lstset\endcsname
\fi

\ifcsname DeclareNewFloatType\endcsname
\DeclareNewFloatType{listing}{}
\floatname{listing}{Program}
\fi

\DeclareMathOperator{\pr}{\normalfont{\textsf{P}}}

\DeclareMathOperator{\var}{\normalfont{Var}}

\DeclareMathOperator{\bin}{\normalfont{Binomial}}

\def\O{\mathcal{O}}

\newcommand{\abs}[1]{\left|{#1}\right|}

\newcommand{\eqdef}{\stackrel{\mathrm{def}}{=}}

\providecommand{\ceiltwo}[1]{\stackrel{2}{\lceil} #1 \stackrel{2}{\rceil} }

\newcommand{\ma}[1]{\underline{#1}}
\newcommand{\un}[1]{\overline{#1}}

\DeclareDocumentCommand{\nlog}{o o m}{\ln_{#1}^{#2}{\hspace{-0.2em}\left(#3\right)}}
\DeclareDocumentCommand{\xlog}{o o m}{\log_{#1}^{#2}{\hspace{-0.2em}\left(#3\right)}}

\DeclareMathOperator{\err}{\normalfont{Err}}
\DeclareMathOperator{\rank}{\normalfont{R}}
\newcommand*{\xhat}[1]{#1\kern-0.8em\hat{\phantom{#1}}}
\DeclareMathOperator{\estrank}{\normalfont{\xhat{R}}}
\DeclareMathOperator{\spa}{\texttt{SPACE}}

\DeclareMathOperator{\lapprox}{\lesssim}
\DeclareMathOperator{\gapprox}{\gtrsim}

\DeclareMathOperator{\mar}{\mathcal{M}}
\DeclareMathOperator{\comp}{\mathcal{P}}

\newcounter{notation}

\ifx\citet\undefined\else

\DeclareDelimFormat[textcite]{multinamedelim}{\addcomma\space}
\DeclareDelimFormat[textcite]{finalnamedelim}{\space and~}

\fi

\newcommand{\pvnote}[1]{\relax}
\newcommand{\tdnote}[1]{\relax}

\title{Relative Error Streaming Quantiles with Seamless Mergeability via Adaptive Compactors\thanks{Supported by the ERC CZ project LL2406 of the Czech Ministry of Education, Youth and Sports, by the grant SVV–2025–260822
		and by Center for Foundations of Modern Computer Science (Charles Univ.\ project UNCE 24/SCI/008).}}
\date{}
\author{Tomáš Domes
	and Pavel Veselý\,\orcidlink{0000-0003-1169-7934}}
\affil{Computer Science Institute of Charles University, Czech Republic\\
	\texttt{\{domestomas,vesely\}@iuuk.mff.cuni.cz}
}

\begin{document}

\maketitle

\noindent Quantile summaries provide a scalable way to estimate the distribution of individual attributes in large datasets that are often distributed across multiple machines or generated by sensor networks.  ReqSketch~\cite{ReqSketch} 
is currently the most space-efficient summary with two key properties: \emph{relative error guarantees}, offering increasingly higher accuracy towards the distribution’s tails, and \emph{mergeability}, allowing distributed or parallel processing of datasets. Due to these features and its simple algorithm design, ReqSketch has been adopted in practice, via implementation in the Apache DataSketches library. However, the proof of mergeability in~\cite{ReqSketch} is overly complicated, requiring an intricate charging argument and complex variance analysis.

In this paper, we provide a refined version of ReqSketch, by developing so-called \emph{adaptive compactors}. This enables a significantly simplified proof of relative error guarantees in the most general mergeability setting, while
retaining the original space bound, update time, and algorithmic simplicity.
%
Moreover, the adaptivity of our sketch, together with the proof technique, yields near-optimal space bounds in specific scenarios -- particularly when merging sketches of comparable size.

\section{Introduction}
\label{intro}

Quantile summaries, or sketches, are data structures of sublinear size, typically just polylogarithmic in the input length,
allowing to efficiently process large streaming or distributed datasets, and 
approximate their distribution, by estimating ranks and quantiles. 
Here, the rank of an item $y$ with respect to $\mathcal{S}$, denoted $\rank(y, \mathcal{S})$,
is the number of items in $\mathcal{S}$ that are smaller than or equal to $y$,
and quantiles are essentially the inverse of ranks, i.e., a $\varphi$-quantile for $\varphi\in [0, 1]$ is the $\lceil \varphi\cdot |\mathcal{S}|\rceil$-th smallest input item.
Quantile summaries are essential for many applications in network monitoring, e.g.,
tracking latencies~\cite{MassonRL19,tene2015_latency_talk} or round trip times~\cite{CormodeKMS05}, or in sensor networks~\cite{q_digest,HeCCW15}.

Many of these applications require two main properties from the summaries that are both challenging to achieve while preserving high space efficiency:
First, the quantile summary should provide an accurate distribution approximation, especially for the tails. 
In particular, network latencies are typically heavily long-tailed~\cite{MassonRL19},
and understanding the tail requires increasingly more accurate estimates of the 95-th, 99-th, 99.5-th, etc., percentiles. 
The notion of \emph{relative error} captures this scenario by requiring that the rank error for an item of rank $R$ is at most $\varepsilon\cdot R$,
for an accuracy parameter $\varepsilon > 0$;
equivalently, for a quantile query $\varphi\in [0, 1]$, the summary should return an $\hat{\varphi}$-quantile, for  $\hat{\varphi} = (1 \pm \varepsilon)\cdot \varphi$.
While this yields lower error for items of small rank, high accuracy for the other tail, as desired in monitoring latencies, can be achieved by simply flipping the comparator.
A long line of work~\cite{GuptaZ03,CormodeKMS05,q_digest_new,ZhangLXKW06,zhang,ReqSketch} studied relative-error quantile summaries,
and recently Gribelyuk et al.~\cite{GribelyukSWY25} developed a randomized streaming algorithm using near-optimal space of $\widetilde{O}(\varepsilon^{-1}\cdot \log \varepsilon N)$\footnote{
	Notation $\widetilde{O}(f)$ hides factors poly-logarithmic in $f$, and the dependency on $\delta$, the probability of a too large error.  
}.

Second, it is in many cases not feasible to process the dataset in one pass
as it may be distributed across many machines or consist of observations aggregated by sensors~\cite{q_digest}.
To allow for distributed or parallel processing of large datasets,
the sketch should come with a merge operation that takes two sketches representing datasets $\mathcal{S}_1$ and $\mathcal{S}_2$ and outputs one sketch representing the multiset union of $\mathcal{S}_1$ and $\mathcal{S}_2$.
A quintessential property of summaries is \emph{full mergeability}~\cite{agarwal2013mergeable},
stating that even when the sketch is built by \emph{any} sequence of pairwise merge operations on single-item sketches obtained from a set $\mathcal{S}$,
its error should be bounded, ideally achieving the same accuracy-space trade-off as in the streaming model. 
The mergeability setting is more general, since the streaming model can be seen as a repeated merge with a single-item sketch.

The state-of-the-art relative-error and fully-mergeable quantile summary is ReqSketch~\cite{ReqSketch}
that consists of a sequence of \emph{relative compactors}. These are buffers that randomly subsample the stream in a way that facilitates the relative error; see Section~\ref{technical_overview}.
Due to its overall simplicity, mergeability, and fast update time, compared to widely-used $t$-digest~\cite{ReqSketch_improvements},
ReqSketch has been implemented within the Apache DataSketches library~\cite{datasketches}.
While ReqSketch only achieves a suboptimal space bound of $\O(\varepsilon^{-1}\cdot \log^{1.5} \varepsilon N)$~\cite{ReqSketch},
the recent improvement in~\cite{GribelyukSWY25} 
does not come with any mergeability guarantees (or even a merge operation) and moreover, the new algorithm from~\cite{GribelyukSWY25} is also arguably much more involved than ReqSketch.
Indeed, to the best of our knowledge, there is no implementation, and it seems likely that ReqSketch would still perform better than
the algorithm of~\cite{GribelyukSWY25} for data streams in practice.
Overall, it is still open how to achieve full mergeability for relative-error sketches in space close to
the lower bound of $\Omega(\varepsilon^{-1}\cdot \log \varepsilon N)$ memory words that holds for any sketch, even if computed offline~\cite{ReqSketch}.

The key obstacle to progress appears to be proving mergeability for relative-error sketches.
Indeed, despite the simplicity of the ReqSketch algorithm and its relatively accessible analysis in the streaming setting,
the proof of full mergeability in~\cite{ReqSketch}
involves an intricate charging argument and even more complex bounds on the variance across relative compactors.
The especially tricky part is that the internal parameters of the sketch change as the sketch summarizes a larger and larger input.
This makes it difficult to either improve the space bound for fully-mergeable sketches or generalize the proof for other error guarantees.

Here, we develop \emph{adaptive compactors} that, when used instead of relative compactors in ReqSketch, allow for a significantly simplified proof of full mergeability while preserving its other properties like space requirements, update time, and relative-error guarantees.
The new compactors naturally capture the charging required to bound the error, and in effect make the analysis of the error in the full mergeability setting as accessible as in the streaming model.
Furthermore, the adaptivity of the compactors allows for getting better space bounds in special cases, obtained as a simple corollary of our analysis.
We demonstrate it
by showing that for a reverse-sorted input stream, the required space is optimal with respect to both $\varepsilon$ and $N$.
Nevertheless, the adaptivity makes bounding the space used by the compactors less straightforward compared to ReqSketch;
we present an analysis of space requirements using a suitable potential function, yielding the same bound of $\O(\varepsilon^{-1}\cdot \log^{1.5} \varepsilon N)$ as in \cite{ReqSketch}.

Finally, we confirm the intuition that certain ``nice'' cases of the mergeability setting admit better space bounds than streaming.
Specifically, we prove that our sketch created 
using merge operations organized in an approximately balanced merge tree (i.e., always merging sketches summarizing a similar number of items), admits a space bound of $\O(\varepsilon^{-1}\cdot \log(\varepsilon N)\cdot \sqrt{\log \log N})$, which is only a factor of $\sqrt{\log \log N}$ from the optimum. 

\subsubsection{Organization of the paper}
In the remaining part of \Cref{intro}, we explain the intuition behind adaptive compactors (\Cref{technical_overview}) and review related work (\Cref{prior_work}). We introduce main notation in \Cref{notation}. In \Cref{description}, we describe our sketch in detail, particularly the new compaction algorithm (\Cref{size}) and we state a few crucial observations and invariants of the algorithm. In \Cref{analysis}, we prove the space bound and the error guarantee of our sketch. In \Cref{improvements}, we present two special cases where our sketch performs particularly well.
We conclude in \Cref{discussion} with a discussion of the main open problem.

\subsection{Technical overview}
\label{technical_overview}

The basic sketch design comes from the KLL sketch \cite{KLL}. Namely, our sketch consists of a series of \emph{compactors}, which are essentially buffers, arranged into \emph{levels}. Items from the input stream are added to the level-0 compactor and when any compactor exceeds its \emph{capacity} $C$, we perform the \emph{compaction} operation.
The compaction sorts the items in the compactor (non-increasingly), deletes odd-indexed or even-indexed items with equal probability, and moves the rest to the compactor one level higher. Therefore, each item at level $h$ ``represents'' $2^h$ items of the original stream.
Then the rank estimate for a query item $x$ is simply the sum of weights of all stored items $y\le x$.
Moreover, as the items are sorted and we delete each odd/even-indexed item, each promoted item represents a deleted item of a similar rank. 
Each compaction affects the error of rank estimates for some universe items, namely, by either increasing, or decreasing their rank estimate by $2^h$ with equal probability, where $h$ is the level. This implies that the expected error change is zero, and then the error analysis for any item $x$ mainly boils down to bounding the number of compactions affecting the estimation error for $x$.

In \emph{relative compactors} by Cormode et al. \cite{ReqSketch}, and also in our \emph{adaptive compactors}, the compaction operation is performed only on some $T$ largest items of the compactor, for an even $T$; see \Cref{compaction_general}.
This is based on the observation that if $x$ is smaller than all compacted items, the compaction does not affect the error for $x$.
The size $T$ of the compaction is determined in such a way that the smallest $C/2$ items are never removed from the compactor. This is related to the relative error guarantee as the accuracy of the sketch must be higher for items of smaller rank. Particularly for the smallest $\varepsilon^{-1}$ items the error must equal zero and so we must store them all (and indeed, the level-0 compactor with capacity $C \geq 2\varepsilon^{-1}$ always stores the smallest $\varepsilon^{-1}$ inserted items).

\begin{figure}
\centering
\includegraphics[width=\textwidth]{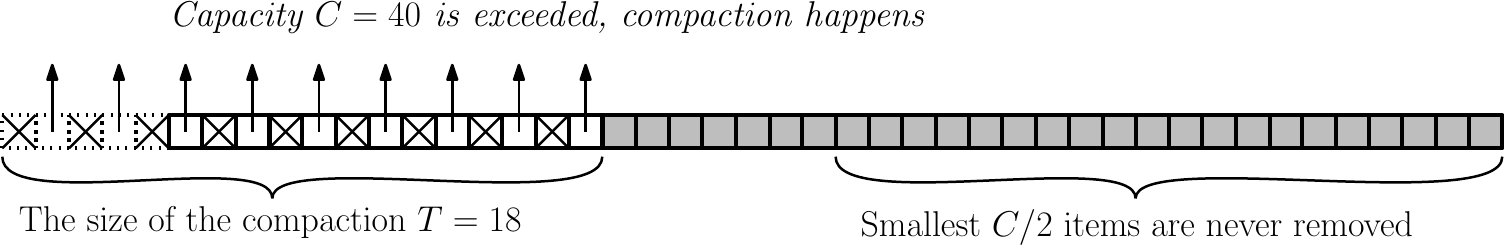}
\caption{The compaction operation of relative and adaptive compactors for a compactor that overflows its capacity (dashed items).
	Items are first sorted from the largest to the smallest.
	The compaction evicts the crossed items from the memory and ``promotes'' the items with arrows to the next level,
	while the gray ones are not involved in the compaction, so they remain in the buffer.
	The size $T$ of the compacted part is, however, computed differently in relative and adaptive compactors.}
\label{compaction_general}
\end{figure}

The main novelty of our work is the choice of the size of the compaction. Cormode et al. \cite{ReqSketch} use a deterministic process called \emph{compaction schedule}, based on derandomizing a suitable geometric distribution. The compactor is split into sections of size $K$, and the schedule is determined by $K$ and by the binary representation of the number $P$ of already performed compactions; namely, if $Z$ is the number of trailing ones of $P$, they have $T = (Z+1)K$. 

To see the motivation behind the compaction schedule, 
let us number the sections starting from zero in such a way that the largest items are in section 0. Similarly, let us index the bits of the binary representation of $P$ from the least significant bit starting from one. The $i$-th bit of $P$ corresponds to the $i$-th section of the compactor, and we say that section $j\ge 1$ is \emph{marked} if the $j$-th bit is set to 1. During a compaction, we remove items from section 0 and from the marked sections $[1, Z]$, where $Z$ is the trailing ones of $P$ (the range can be empty); as in KLL, a half of these removed items promoted to the next level and discarded otherwise. Then we increment $P$, which essentially means marking the previously unmarked section $j$; see \Cref{compaction_relative}.

\begin{figure}
\centering
\includegraphics[width=0.8\textwidth]{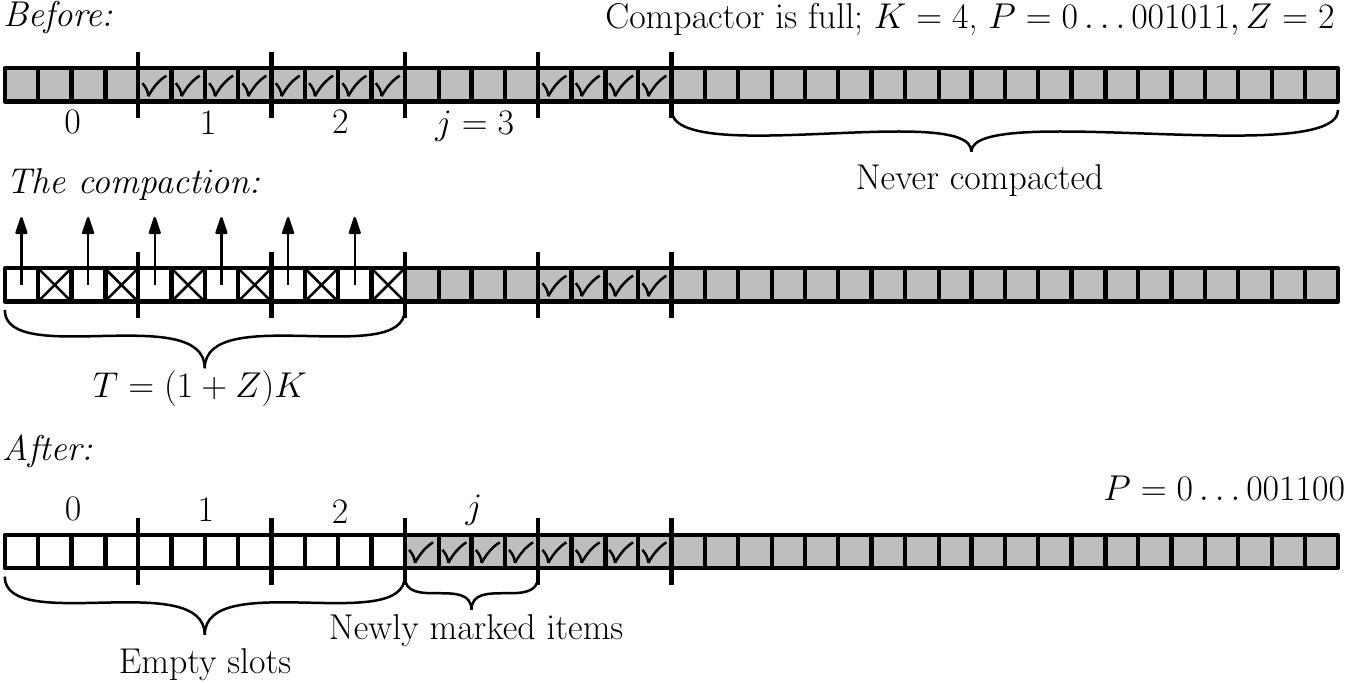}
\caption{The compaction operation in a relative compactor with capacity $C = 40$,
where $P$ is the number of already performed compactions (written in binary), and $Z$ equals the number of trailing ones of $P$.}
\label{compaction_relative}
\end{figure}

At time $t_0$ of marking, all items in section $j$ are of smaller rank than all items removed from the compactor by the compaction. After another $2^j$ compactions at time $t_1$, the marked section $j$ is compacted.
Generally, the items present in section $j$ at time $t_1$ can be different than items there at time $t_0$.
However, due to the sorting before each compaction, all the new items have smaller or equal rank than all the old ones. This implies the key property that \emph{for each compaction, there are some $K$ unique items of smaller rank than all items removed by the compaction which are either later removed by another compaction or remain in the compactor at the end of the algorithm.} This can be used to bound the number of compactions deleting small items, which is essential for the relative error guarantee. We refer to~\cite{ReqSketch} for more details.

Unfortunately, the notion of compaction schedule becomes much less elegant under merge operations. Cormode et al. \cite{ReqSketch} developed a merging algorithm that has two deficiencies:
First, when merging two relative compactors, they take binary OR of their compaction schedules,
which is efficient but the nice intuition behind the compaction schedule is somehow lost and the previously easy-to-see bounds on the number of compactions affecting the error require an involved charging argument (Lemma~6.4 in~\cite{ReqSketch}).
Second, the parameters of each sketch are recomputed multiple times, every time the summarized input size of the sketch gets squared,
and this yields a complicated tight analysis of the variance (e.g., Lemma~6.12 in~\cite{ReqSketch}).

We propose an alternative to the compaction schedule which is a bit less elegant in the streaming setting, but naturally covers full mergeability with much simpler proofs.
The core of our method is to tie the marks not to the sections of the compactor, but directly to the items.
Then we can simply merge two compactors by merging the two sorted arrays -- all the marked items simply stay marked. However, in this model the marked items no longer form contiguous sections of the compactor, due to future item additions.
It is possible to analyze the error without enforcing contiguous sections of marked items, but the resulting compaction algorithm and its analysis are more complicated, requiring the notions of dense and sparse sections with respect to the marking.

Instead, we further modify the algorithm to recover contiguous sections of marked and unmarked items, which yields simpler analysis.
Suppose the lowest ranked item removed during a particular compaction has rank $r_0$ and the items marked after the compaction have rank $r_1 > r_2 > \dots > r_K$. Then a new item with rank $r_x$ arrives such that $r_0 > r_x > r_K$. Such item can be surely marked instead of the item $r_K$. With such remarking, the $K$ items marked for a particular compaction always form a contiguous section in the compactor. To do this remarking efficiently, we do not actually mark the items, but remember for each compaction the smallest-ranked removed item, called the \emph{ghost item}, and the current value of $K$. From this information, we can in linear time recover the marking of items. We perform the ``remarking'' so that the contiguous sections of marked items become aligned, i.e., start on positions divisible by $K$, which further simplifies the analysis and it allows us to have marked sections instead of marked items, similarly as for relative compactors.

Given the marking of sections, determining the size of the compaction is conceptually the same as for relative compactors. We look for the first unmarked section $j$, remove all items to the left of this section, and then mark $j$, as depicted in \Cref{compaction_relative}.

Our strategy preserves the property of the relative compactors that for each compaction, there are some $K$ unique items that are smaller than or equal to those removed by the compaction and that will either be removed by another compaction, or remain in the compactor at the end of the algorithm. Moreover, with adaptive compactors, this property is a simple observation, even considering the change of parameters over time and full mergeability, which are both naturally covered by our algorithm. As mentioned above, proving such a property for relative compactors from~\cite{ReqSketch} under full mergeability requires intricate analysis. The rest of the error analysis is similar to that of ReqSketch in the streaming setting but in our case, it covers full mergeability for free. We partially pay for this simplicity in the analysis of the space bound, which was trivial for relative compactors. For adaptive compactors, we define a suitable potential function which connects the capacity of the compactor with the number $P$ of performed compactions; see \Cref{space_bound}.
However, while relative compactors always use the same space for a fixed input size no matter the input structure,
the advantage of our approach is that it adapts, allowing us to prove near-optimal space bounds in special cases.

To sum up, our analysis of the error under arbitrary merge operations is now comparable to the relatively simple analysis of ReqSketch in the streaming setting,
while the space analysis requires a new potential-function argument.

\subsection{Related work}
\label{prior_work}

Here, we focus on state-of-the-art quantile summaries with various error guarantees, with a particular focus on their mergeability properties.
The space bounds are given in memory words with $\Theta(\log \mathcal{U} + \log N)$ bits, where $N$ is the input size and $\mathcal{U}$ is the universe of input items.
For randomized algorithms, we present the space bounds assuming a constant probability $\delta$ of a too large error for any fixed query.

Besides relative error, there is a large body of work on quantile summaries with the \emph{additive (uniform) error} of $\pm \varepsilon N$,
i.e., without higher accuracy for the tails.
The state-of-the-art additive-error sketch is the randomized comparison-based KLL sketch by Karnin, Lang, and Liberty~\cite{KLL} achieving full mergeability in the optimal space $\Theta(\varepsilon^{-1})$.
For deterministic comparison-based sketches, the optimal space of $\Theta(\varepsilon^{-1}\cdot \log(\varepsilon N))$ is achieved by the Greenwald-Khanna (GK) summary~\cite{GK}; its optimality has only been shown recently~\cite{add_lower} (for a simplified version of the GK summary, see \cite{GK_new}). However, the Greenwald-Khanna summary is not fully mergeable, only ``one-way mergeable'' as discussed in~\cite{agarwal2013mergeable}. An older deterministic algorithm by Manku, Rajagopalan, and Lindsay~\cite{MRL} achieves full mergeability in space $\O(\varepsilon^{-1}\cdot \log^2(\varepsilon N))$.

The non-comparison-based deterministic additive-error $q$-digest by Shrivastava et al. \cite{q_digest} achieves full mergeability in space $\O(\varepsilon^{-1}\cdot \log(\abs{\mathcal{U}}))$, that is, depending on the universe size, which must be known in advance. In this model, Gupta et al.~\cite{GuptaSW-noncomparison_optimal} recently achieved the optimal space $\Theta(\varepsilon^{-1})$ deterministically, using a compressed version of $q$-digest, but they did not prove that their sketch is fully mergeable.

For deterministic comparison-based sketches with the \emph{relative error guarantee}, we have a lower bound of $\Omega(\varepsilon^{-1}\cdot \log^2(\varepsilon N))$ by Cormode and Veselý \cite{add_lower}, while the state-of-the-art sketch by Zhang and Wang \cite{zhang} achieves space $\O(\varepsilon^{-1}\cdot\log^3(\varepsilon N))$; however, it is also not fully mergeable.

There is a modified version of the aforementioned $q$-digest by Cormode et al.~\cite{q_digest_new} achieving relative error in space $\O(\varepsilon^{-1}\cdot\log(\varepsilon N)\cdot \log(\mathcal{U}))$. Like the original $q$-digest, it is deterministic, non-comparison-based, fully mergeable, and requires the prior knowledge of the universe $\mathcal{U}$.

In the randomized comparison-based setting, the state-of-the-art relative-error algorithms are ReqSketch by Cormode et al. \cite{ReqSketch} and the recent work of Gribelyuk et al. \cite{GribelyukSWY25}. ReqSketch is fully mergeable but requires $\O(\varepsilon^{-1}\cdot \log^{1.5}(\varepsilon N))$ space, still by $\sqrt{\log(\varepsilon N)}$ larger than the lower bound of $\Omega(\varepsilon^{-1}\cdot \log(\varepsilon N))$ that holds even for non-comparison-based randomized algorithms~\cite{ReqSketch}. The work of Gribelyuk et al.~\cite{GribelyukSWY25} almost closes this gap by designing a streaming algorithm using space
$\O(\varepsilon^{-1}\cdot \log(\varepsilon N)\cdot \log \varepsilon^{-1}\cdot (\log \log N + \log \varepsilon^{-1}))$. However, the sketch was not proven to be mergeable; in fact, no merge operation was designed to preserve properties required in the analysis.

Finally, besides algorithms designed with theoretical guarantees in mind, it is worth mentioning that many practitioners actually use $t$-digest~\cite{tdigest}, which is fully mergeable and usually highly-accurate in practice, aiming at uniform or relative error depending on its parameter setup, but providing no theoretical worst-case guarantees. Indeed, Cormode et al.~\cite{ReqSketch_improvements} demonstrated that for adversarially constructed inputs or even samples from certain distributions, the error of $t$-digest can be almost arbitrarily large.

\section{Preliminaries}
\label{notation}
Here, we outline notation and terminology used throughout the paper:
\begin{itemize}
\item By \emph{item}, we always mean an arbitrary item from a universe $\mathcal{U}$ with a total order.
For simplicity, we assume in the whole paper that all the input items are different.
However, all the definitions can be extended to support equal items and all the
algorithms and proofs are correct even without this assumption.

\item An input $\mathcal{S}$ consists of items from $\mathcal{U}$. We denote the input size by $\abs{\mathcal{S}} = N$.

\item For any set of items $S$, let $\rank(y, S)$ be the rank of item $y$ in $S$, and let $\rank(y) = \rank(y, \mathcal{S})$ be the rank of $y$ in the input.

\item We use $\estrank(y)$ for the answer returned by the sketch for a rank query $y$.

\item The error of a rank query $y$ is $\err(y) = \estrank(y) - \rank(y)$.

\item By $\log x$ we always mean $\log_2 x$ and by $\ln x$ the natural logarithm $\log_e x$.

\item We denote the ranges of arrays as mathematical intervals, thus $)$ means exclusion and $]$ inclusion. For example $A[3, 5]$ are elements $A[3], A[4]$ and $A[5]$ while $A[3, 5)$ means only $A[3]$ and $A[4]$.

\end{itemize}

\section{Description of the Algorithm}
\label{description}

The high-level design of the sketch is analogous to ReqSketch \cite{ReqSketch}, namely that it consists of compactors arranged into levels.
However, our adaptive compactors are more independent of each other, having their own parameters, and use a more flexible strategy for performing compactions. In fact, the most significant change is in determining the size of the compaction (\Cref{size}).

\subsubsection{Basic sketch design}
The sketch consists of a sequence of $H$ \emph{adaptive compactors} indexed from 0 to $H - 1$. We imagine the compactors as arranged in \emph{levels} -- the compactor at level 0 is at the bottom, the compactor at level $H-1$ at the top. 

An adaptive compactor has some internal state, a \emph{capacity} $C$, an input stream $I$, an output stream $O$, and contains a set $B$ of items (the buffer).
The input of the compactor on level 0 is the input of the sketch and the input of any compactor on level $\ell > 0$ is the output of the compactor on level $\ell - 1$. 

Any set of items received from the input stream is simply added to $B$. Whenever the size of $B$ reaches the capacity $C$ (we say that the compactor is \emph{full}), the compactor performs a \emph{compaction operation}, which removes some items from $B$, a half of them are evicted from memory and the other half are sent to the output stream (this can trigger a compaction operation on the next level).
We stress that the items are sent in batches -- if a compaction operation outputs a set of items, all the items are atomically added to the next-level compactor before any other compaction happens.

Any time the compactor on the highest level $H-1$ performs a compaction, the number of levels $H$ is increased by one, and a new compactor on the highest level is created to receive the output stream of the previously highest compactor.

Apart from the sequence of compactors, the sketch consists of the error bound $0 < \varepsilon < 1$ and the failure probability $0 < \delta \leq \frac{1}{8}$.

\subsubsection{Merging two sketches and stream updates}
The merge operation is defined on two sketches $\mathcal{S}_1$ and $\mathcal{S}_2$ with parameters $\varepsilon_1 = \varepsilon_2$ and $\delta_1 = \delta_2$. Without loss of generality (w.l.o.g.), assume that $H_1 \geq H_2$. The parameters $\varepsilon$ and $\delta$ are kept unchanged and for each level $h < H_2$ the compactor on level $h$ of the new sketch is given by a merge of the two compactors on level $h$ from sketches $\mathcal{S}_1, \mathcal{S}_2$. We describe the merge operation on compactors in \Cref{the_compactor}. For $h \geq H_2$ the resulting compactor is simply the compactor on level $h$ from the sketch $\mathcal{S}_1$. After the merge, we go through the compactors bottom-up and perform a compaction on each full compactor.

The update operation takes the next item from the input stream of the algorithm and sends it to the input stream of the level-0 compactor. Note that this operation can be viewed as a merge with a trivial sketch containing one element.

\subsubsection{Answering queries}
For the rank query $y$ the sketch returns value $\estrank(y)$, which is defined as
$$
\estrank(y) \eqdef \sum_{h = 0}^{H-1} 2^h \rank(y, B_h).
$$
To implement queries more efficiently, we sort items from all the compactors together with their weights and precompute weighted prefix sums. Then we answer a rank query by a single binary search over the sorted items (and one lookup to the prefix sums) and a quantile query by a single binary search over the prefix sums (and one lookup to the items).

\subsection{Adaptive compactor}
\label{the_compactor}
The adaptive compactor consists of the \emph{buffer} $B$ of items, \emph{capacity} $C$, \emph{section length} $K$, and a stack $M$ of \emph{markers}. We often use subscript to denote the level of the compactor (thus $C_h$ is capacity of the level-$h$ compactor).

Each marker is a tuple (\emph{length}, \emph{ghost item}) where \emph{length} is an integer and \emph{ghost item} is a copy of some item previously removed from $B$.
The comparison of two markers $(\ell_1, g_1), (\ell_2, g_2) \in M$ is given by the comparison of the ghost items $g_1, g_2$ and the comparison of a marker $m = (\ell, g)$ and item $y$ is given by the comparison of ghost item $g$ and item $y$ (thus we can say that ``marker $m$ is smaller than item $y$'').

For the analysis, we assume that both $B$ and $M$ are always sorted\footnote{In practice we keep $M$ sorted and we sort $B$ before every compaction.} with the largest item on position $B[0]$ and the largest marker on the top of $M$.

The initial setting of parameters is the same regardless of the level. At start, we have naturally $B_0 = \emptyset$ and $M_0 = \emptyset$. Let us denote the closest larger power of two by $\ceiltwo{\!x\!} \,\eqdef 2^{\lceil \log x \rceil} $.
The initial values of parameters $C$ and $K$ are chosen as follows:
$$
K_0 = \ceiltwo{ \max\left( \varepsilon^{-1}\sqrt{\ln \delta^{-1}},\, 4 \ln \delta^{-1}\right) };\;\;\;
C_0 = 8 K_0.
$$
The definition may seem cryptic, but we later see that $C_0$ and $K_0$ are defined exactly such that invariants \ref{pow2}, \ref{smallK}, \ref{Clog} and \ref{Cinv} hold (defined in \cref{size}) and $C_0$ is the smallest possible.

\subsubsection{Compaction}
The compaction operation is performed on a single compactor and it happens only when the compactor is full, i.e., $\abs{B} \geq C$. First, the \emph{size} $T$ of the compaction is determined, as described in \Cref{size}. If $T$ is odd, we temporarily remove the largest item from $B$, store it and decrease $T$ by one (and we return the item back to $B$ after the compaction is performed).
Then, we remove the largest $T$ items from $B$ and with probability $1/2$, the odd-indexed removed items are sent to the output stream of the compactor, and in the other case, the even-indexed are sent to the output. The other half of the largest $T$ items is evicted from memory.

Note that the operation itself is defined exactly as in the relative compactor of ReqSketch \cite{ReqSketch}. However, the key difference is in determining the size of the compaction described in \Cref{size}.

\subsubsection{Merging two adaptive compactors}
The merge of two compactors is straightforward: We merge the lists $B$ and $M$ and inherit the parameters $C$ and $K$ from the compactor with larger $C$ and smaller $K$ (this is needed for invariant \ref{monotonicity} stating that $K$ never increases and $C$ never decreases). We note that the compactor with larger $C$ has always smaller $K$, as by invariant \ref{Cinv} the expression $KC$ is constant.
 
Thus, the result of the merge of two adaptive compactors $\mathcal{C}_1 (B_1, C_1, K_1, M_1), \mathcal{C}_2 (B_2, C_2, K_2, M_2)$ where $C_1 \geq C_2$ is defined by $B = B_1 \cup B_2$, $C = C_1$, $K = K_1$, and $M = M_1 \cup M_2$.

\subsubsection{Sections}
\label{sections}

For the sake of later analysis, we divide $B$ into \emph{sections} of length $K$, where each section is a range $B\bigl[\abs{B}-(i+1)K, \abs{B}-iK\bigr)$ for some integer $i \geq 0$.

Our algorithm maintains the invariants that $C$ is always a multiple of $2K$ \ref{CdivK} and that $C \geq 8K$ \ref{smallK}. This allows for the following notation: We index the sections by consecutive increasing integers such that the index of the last section is $C/K-1$. Thus, when $|B| = C$, the largest $K$ items lay exactly in section 0 and when $|B| > C$, there are some items in negative-indexed sections (and those items always participate in the compaction).
We denote the $i$-th section of $B$ by $S[i]$ and analogously for ranges of sections (thus $S[i, j]$ are all sections from $i$-th to $j$-th). Let us name the range $S[-\infty, C/2K)$ as the \emph{left part} of $B$ and the range $S[C/2K, C/K)$ as the \emph{right part} of $B$. Moreover, let $S[-\infty, 0]$ be the \emph{tail} of $B$, and let the \emph{head} consist of the last 2 sections of $B$ (see \Cref{parts}). Note that as $C \geq 8K$ \ref{smallK}, the tail is contained in the left part and the head is contained in the right part.

\begin{figure}
\centering
\includegraphics[width=\textwidth]{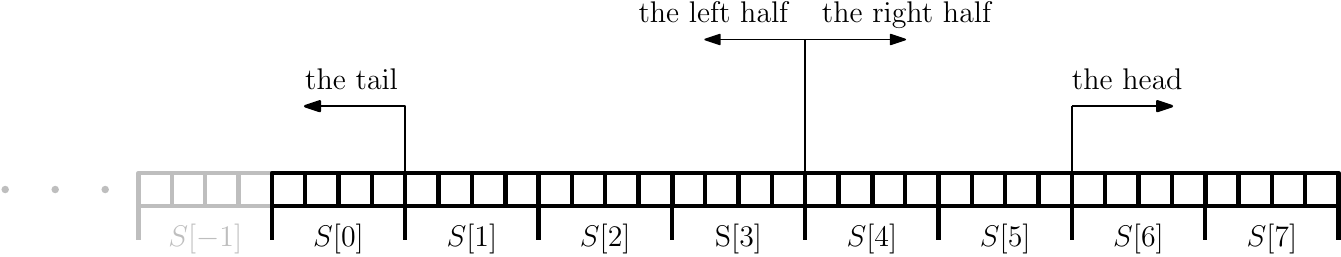}
\caption{Sections of a buffer with $K = 4$ and $C = 32$}
\label{parts}
\end{figure}

\subsubsection{Marking}

We use the markers from definition of an adaptive compactor to mark the items in $B$. For any marker $m$ with ghost item $g$ and length $\ell$ we mark some $\ell$ items of $B$ that are smaller than $g$ such that each item of $B$ is marked by at most one marker.

In \Cref{size} we explain our algorithm which for each compaction creates a marker of length $K$ where the ghost item is the smallest of the items removed during the compaction.
Thus, as we formally prove in \Cref{coins}, for each compaction there are $K$ unique ``small'' marked items and this can be used to bound the number of compactions affecting the error of a given item $y$ (see \Cref{mbound} and \Cref{final_error_estimate}).

Let us now formally define the marking, before we dive into our algorithm.

\begin{defn}[Marking]

For given buffer $B$ and list of markers $M$, \emph{marking} of $B$ by $M$ is defined as a function $\mar$ from markers to sets of items in $B$ such that:
\begin{itemize}
\item for each $m \in M$ and $y \in \mar(m): y \leq m$
\item $\abs{\mar(\ell, g)} = \ell$
\item for different markers $m_1, m_2 \in M: \mar(m_1) \cap \mar(m_2) = \emptyset$
\end{itemize}
For a given marking $\mar$, the items from $B$ that are contained in some set of the image of $\mar$ are \emph{marked} and the rest of the items are \emph{unmarked}.
\end{defn}

\begin{defn}[Canonical marking]
For sorted buffer $B$, section length $K$ and sorted stack of markers $M$, canonical marking is obtained by the following greedy procedure: Remove the largest marker $m = (\ell, g)$ from $M$, find the first section $s = B[i, i+K)$ such that all items in $s$ are unmarked and smaller than $m$ and set $\mar(m) = B[i, i+\ell)$. Repeat while there are some markers in $M$ left.
\end{defn}

We say that a section of $B$ is \emph{full} if it contains $K$ items, and it is \emph{marked} if it contains marked items only; otherwise, it is \emph{unmarked} and contains unmarked items only, by the definition of the marking algorithm.
We maintain an invariant that for each marker $m = (\ell, g)$, the length $\ell$ is a multiple of $K$ \ref{longmark}. Thus, in the canonical marking, each section is either marked or unmarked.

\subsection{The size of the compaction}
\label{size}
In this section, we present our algorithm for determining the size $T$ of the compaction which uses the notion of marking defined above. This algorithm is the key difference between the relative compactors by Cormode et al. \cite{ReqSketch} and our adaptive compactor.

The main idea is to find the leftmost unmarked section $s$ to the right of the tail, compact all the items from the previous sections (together with the markers that marked them) and create a new marker that marks $s$. The pseudocode is shown in \Cref{alg:compaction-size}.

The core of the algorithm is on \crefrange{code:basic_start}{code:basic_end}. We go through the sections of $B$ and for each section $s$ we check whether it is marked in the canonical marking. When we find an unmarked section in the left part of $B$ (but not on the tail), we compact all items before $s$ (\cref{code:unmarked}) and create a new marker to mark $s$. This is called a \emph{standard compaction}. Note that we do not mark the items explicitly, as it is in fact not needed, they are only marked for the sake of the analysis.

The \emph{special compaction} (\cref{code:special_compaction}) corresponds to the case when all sections in the left part that are not on the tail are marked. In this case, we compact the whole left part and mark the first unmarked section of the right part. It is possible that in this situation, not all of the items marked by the last-removed marker get removed (as part of them may lay in the right part). If so, we increase the length of the new marker accordingly (\cref{code:bigmarker}).

\Crefrange{code:param_start}{code:param_end} take care of change of parameters of the compactor -- if the number of unmarked sections is too small, we double $C$ and halve $K$. \Cref{code:Ksmall_start,code:Ksmall_end} implement a naive algorithm that we use when $K$ becomes 1.

\begin{algorithm}
\caption{Determining the size of a compaction}
\label{alg:compaction-size}
\begin{algorithmic}[1]

	\Require $|B| \geq C$ \Comment{this procedure runs only if buffer is full}
	\Require $B$ is sorted with largest item $B[0]$, $M$ is sorted with largest item at the top

	\If{$K = 1$}\label{code:Ksmall_start}
		\State \Return $|B| - C/2$\label{code:Ksmall_end}\Comment{naive compaction}
	\EndIf

	\State $T \gets |B| \mod K$\label{code:basic_start}

	\While{\textbf{true}}
		\If{$M$ is empty \textbf{or} $B[T] > M$.peek()} \Comment{$B[T, T+K)$ are unmarked}\label{code:unmarked}
			\If{$T < |B| - C + K$} \label{code:smallT} \Comment{$T$ is too small}
				\State $T \gets T + K$ \label{code:incT1}
			\Else \label{code:standard_compaction}\Comment{standard compaction}
				\State push $(K, B[T-1])$ to $M$
				\State \Return $T$
			\EndIf
		\Else \label{code:marked}
			\State $(length,ghost) \gets M$.pop() \Comment{$B[T : T+length)$ are marked}
			\State $T \gets T + length$ \label{code:basic_end}
			\State $overlap \gets T - (|B| - C/2)$
			\If{$overlap \geq 0$} \label{code:special_compaction}\Comment{special compaction}
				\State $T \gets T - overlap$
				\State insert $(K+overlap, B[T-1])$ to $M$ \label{code:bigmarker}\Comment{keep $M$ sorted}
				\State $unmarked \gets C/2 -{}$sum of lengths of all markers in $M$
				\If{$unmarked < 2K$}\label{code:param_start} \Comment{in fact we have $unmarked = K$}
					\State $K \gets K/2$\label{code:setK}
					\State $C \gets 2C$\label{code:param_end}
				\EndIf
			\EndIf
		\EndIf
	\EndWhile
\end{algorithmic}
\end{algorithm}

Our algorithm maintains the following invariants:

\begin{enumerate}[label=(I\arabic*)]

\item \label{pow2} \textbf{\boldmath{}The parameters $K$ and $C$ are both always powers of two} by \cref{code:setK,code:param_end} and by the initial setting of $K_0, C_0$.

\item \label{monotonicity} \textbf{\boldmath{}The parameter $K$ is nonincreasing in time and the parameter $C$ is nondecreasing in time}\footnote{Note that this holds also after merge, as we always pick the parameters from the sketch with smaller $K$ and larger $C$.} by \cref{code:setK,code:param_end}.

\item \label{smallK} \textbf{\boldmath{}We have always $C \geq 8K$}
by the initial setting of $C_0$ and $K_0$ and by invariant \ref{monotonicity}.

\item \label{Clog} \textbf{\boldmath{}We have always $C \geq 2^5\ln \delta^{-1}$} by the initial setting of $C_0$ and $K_0$ and by invariant \ref{monotonicity}.

\item \label{CdivK} \textbf{\boldmath{}$C$ is always a multiple of $2K$} by invariants \ref{pow2} and \ref{smallK}.

\item \label{Cinv} \textbf{\boldmath{}The value of the expression $KC$ never changes and we have $KC \geq 2^3 \varepsilon^{-2} \ln \delta^{-1}$ and $KC \in \Theta(\varepsilon^{-2}\ln \delta^{-1} + \ln^2 \delta^{-1})$.} This is by  \cref{code:setK,code:param_end} and by the initial setting of $C_0$ and $K_0$.

\item \label{removeK} \textbf{\boldmath{}After the compaction, the tail contains at most one item and the number of removed items is at least $K$},
as each compaction removes all items from the tail (by \cref{code:smallT}) and after the compaction one temporarily removed item can be added back to $B$ (as explained in the description of the compaction operation) only if $T$ was odd and consequently the number of items of $B$ was at least $C+1$.

\item \label{protected} \textbf{\boldmath{}No compaction removes items from the right part of $B$.}
This is clear for naive and special compaction. For standard compaction, consider when $T$ attained its final value. If it was on \cref{code:incT1}, at least $C-2K$ items remain in $B$ and by invariant \ref{smallK} this is larger than $C/2$. If it was on \cref{code:basic_end}, we immediately checked this constraint on \cref{code:special_compaction} and if it was violated we would perform a special compaction.

\item \label{longmark} \textbf{\boldmath{}The length of each marker is always divisible by $K$.}
By induction over the steps of the algorithm. As $K$ is initially a power of two, dividing $K$ by $2$ on \cref{code:setK} does not violate the invariant. Creating a new marker does not violate the invariant if and only if the \emph{overlap} is always divisible by $K$. This is indeed true as $C/2$ is always divisible by $K$ and $|B|-T$ is always divisible by $K$ due to \cref{code:basic_start} and by induction hypothesis (\cref{code:basic_end}).

\item \label{contiguous_sections} \textbf{\boldmath{}All the compacted items not lying on the tail are marked.}
This is because when we find unmarked section with positive index (thus not on the tail), we stop the compaction before this section (\cref{code:standard_compaction}).

\item \label{whitehead} \textbf{\boldmath{}The canonical marking of $B$ by $M$ always exists and leaves the head unmarked.}
This is clearly true for empty $M$.

In case of a standard compaction, the new marker marks a previously unmarked section of the left part and all the old markers mark the same sections as before the compaction.

In case of a special compaction, the length of the new marker is equal to the number of items of the right part marked by the last removed marker plus $K$. Thus, the number of marked items of the right part increases by $K$. Since the new marker is larger than all the items present in $B$ after the compaction, the new marking of the right part differs only on the first previously unmarked section from the old one.
If this section was on the head, the invariant is broken and we immediately restore it on \cref{code:setK} by halving the section size and thus halving the head size.

\end{enumerate}

We need for the invariants to hold even after merging two compactors. This is easy to check for each but the last one. The validity of invariant \ref{whitehead} under the merge operation is a direct corollary of the following observation which we use again later in the analysis of the space bound in \Cref{space_bound}.
Informally, the observation states that after a merge, none of the markers moves to the right.

\begin{obs}\label{markersRange}
For a marker $m$ let its rank in $B$ be defined as  $\rank(m, B) \eqdef \rank(z, B)$ where $z$ is the largest (thus least-indexed) item marked by $m$ in the canonical marking.

For each marker $m$, its rank after a merge is at least its rank before a merge.
\end{obs}

\begin{proof}
Consider a merge of $\mathcal{C}_1 (B_1, C_1, K_1, M_1)$ and $\mathcal{C}_2 (B_2, C_2, K_2, M_2)$ resulting in $\mathcal{C} (B, C, K, M)$. We assume that the canonical marking of $B_1$ by $M_1$ and canonical marking of $B_2$ by $M_2$ both exist. We prove the statement by induction over the number of markers after the merge.

When $|M| = 1$ the statement clearly holds, as the new section length is a divisor of the old one and for each item $x$ from any of the buffers $B_1, B_2$ its rank does not decrease after a merge.

For the induction step, let $m = (\ell, g)$ be the smallest marker in $M$ and w.l.o.g., assume that $m \in M_1$. Let us remove $m$ from $M_1$ and remove the smallest $\rank(m, B_1)$ items from $B_1$. Note that $\rank(m, B_1)$ is a multiple of $K$ and also that the canonical marking of $B_1$ by $M_1$ still exists. By the induction hypothesis, the statement of the lemma now holds.

Let us now add all the removed items back to $B_1$ and $B$. After this operation, rank of all the markers in $M_1$ and $M$ (w.r.t. their buffers) increases by $\rank(m, B_1)$, as in both cases all the added items are smaller than all the markers. Thus, the statement of the lemma still holds and the smallest $\rank(m, B_1)$ items of $B$ are unmarked. Let us add $m$ to $M_1$ and to $M$. This can not influence which items are marked by the rest of the markers, as $m$ is the smallest one. Thus, the smallest $\rank(m, B_1)$ items of $M$ are still unmarked by the rest of the markers and if $m$ does not mark any larger items, it surely marks the largest $\ell$ of the smallest $\rank(m, B_1)$ items of $B$. We conclude that $\rank(m, B) \geq \rank(m, B_1)$ and the statement of the lemma holds.
\end{proof}

Following observation is crucial in bounding the error of the sketch (we use it in \Cref{mbound} and \Cref{final_error_estimate}).

\begin{obs}\label{coins}
For a compaction $\comp$, let $K_{\comp}$ be the value of $K$ when the compaction happens and let $y_{\comp}$ be the smallest item removed from $B$ by compaction $\comp$. For each compaction $\comp$ there exists a set $W_{\comp}$ of items of the input stream of the compactor such that:
\begin{enumerate}[label=(R\arabic*)]
\item \label{coins_k} $\abs{W_{\comp}} = K_{\comp}$,
\item \label{smaller} all the items of $W_{\comp}$ are smaller or equal to $y_{\comp}$, and
\item \label{coins_distinct} for any two distinct compactions $\comp, \mathcal{Q}$ the sets $W_{\comp}$ and $W_{\mathcal{Q}}$ are disjoint.
\end{enumerate}
\end{obs}

\begin{proof}
We prove the lemma by distributing coins. For a compaction $\comp$, we create $K_{\comp}$ $\comp$-coins and distribute them as described below, ensuring that each item from $W_{\comp}$ receives exactly one $\comp$-coin and no coins from other compactions.

If $K_{\comp} = 1$ we send  one $\comp$-coin to item $y_{\comp}$ (which is removed from $B$ during compaction $\comp$). 
For $K_{\comp} > 1$, we send $K_{\comp}$ $\comp$-coins to marker  $\ell$ created during compaction $\comp$. If marker $\ell$ is later deleted during compaction $\comp'$, we send one coin from  $\ell$ to each item marked by  $\ell$ in the optimal marking and removed from $B$ during compaction $\comp'$;
if $\comp'$ is special and not all of these items are removed, we send the rest of the coins to the newly created marker. Observe that the amount of coins on a marker is always exactly its length. Moreover, for each $\comp$-coin on marker $m = (g, l)$ we have always $g \leq y_{\comp}$, showing that property \ref{smaller} holds.
By invariant \ref{whitehead} the optimal marking of $B$ by $M$ always exists and thus, we can distribute the coins remaining in $M$ of the final sketch to the items remaining in the buffer of the final sketch. Therefore, this process of distributing coins maintains all properties~\ref{coins_k}, \ref{smaller}, and~\ref{coins_distinct}.
\end{proof}

\begin{rem}[Time complexity]
Note that apart from sorting $B$, the time complexity of any compaction, including \Cref{alg:compaction-size}, is linear in the number of items removed from $B$.
If we slightly modify the algorithm and let the compaction happen only when $|B| \geq 2C$, we get amortized update time $\O(\log C)$ per item removed from $B$, due to sorting $B$; note that the sketch works in the comparison-based model. Considering the time over all levels, since a half of the removed items from a compactor is deleted forever, the amortized update time of the whole sketch is only $\O(\log C_{\max})$ where $C_{\max}$ is the largest capacity achieved by a compactor of the final sketch. In \Cref{space_bound}, we prove that on any level, $C \in \O(\sqrt{\log(\varepsilon N)} \cdot (\varepsilon^{-1} \sqrt{\ln \delta^{-1}} + \ln \delta^{-1}))$, obtaining amortized update time of $\O(\log \varepsilon^{-1} + \log \log (\varepsilon N) + \log \delta^{-1})$, which is the same as for ReqSketch~\cite{ReqSketch}.
\end{rem}

\section{Analysis}
\label{analysis}
In this section, we give analysis of the sketch under full mergeability. We prove the space bound in \Cref{space_bound} (\Cref{lem:spaceBound}) and the error guarantee in \Cref{error_bound} (\Cref{final_error_estimate}). These together give the main theorem, which is analogous to Theorem 1 in \cite{ReqSketch}.

\begin{thm}\label{mainthm}
For any parameters $0 < \delta < 1/8$ and $0 < \varepsilon < 1$ there is a randomized, comparison-based, fully mergeable streaming algorithm that, when processing an input consisting of $N$ items from a totally-ordered universe $\mathcal{U}$, produces a summary $S$ satisfying the following property.
Given $S$, for any $y \in \mathcal{U}$ one can derive an estimate $\estrank(y)$ of $\rank(y)$ such that
$$
\pr \biggl[ \abs{\estrank(y) - \rank(y)} > \varepsilon \rank(y) \biggr] < \delta\,,
$$
where the probability is over the internal randomness of the algorithm. The size of $S$ is 
$$
\O\left(\log^{1.5}(\varepsilon N) \cdot (\varepsilon^{-1}\cdot \sqrt{\ln \delta^{-1}} + \ln \delta^{-1}) \right).
$$
\end{thm}

\begin{rem}
	Assuming that $\varepsilon^{-1} \geq \sqrt{\ln \delta^{-1}}$, which is typically\footnote{For the standard choice $\varepsilon = 0.01$ we can set the failure probability as low as $2^{-1000}$ without violating the inequality. For large $\varepsilon = 0.1$ we can still have $\delta = 2^{-32}$.} the case for many applications, the space bound becomes simply
$$
\O\left(\log^{1.5}(\varepsilon N) \cdot \varepsilon^{-1}\cdot \sqrt{\ln \delta^{-1}} \right),
$$
which is exactly the same as for ReqSketch~\cite{ReqSketch}.

Furthermore, similarly as for ReqSketch, the sketch also satisfies the space bound of
	$\O\left(\log(\varepsilon N) \cdot \varepsilon^{-2} \ln \delta^{-1}\right)$,
	which is only better than the bound in \Cref{mainthm} for $\varepsilon^{-1}\cdot \sqrt{\ln \delta^{-1}} \le \sqrt{\log(\varepsilon N)}$, i.e.,
	for a very large input size $N$, exponential in $\varepsilon^{-2}\cdot \ln \delta^{-1}$.
\end{rem}

\subsection{The space bound}
\label{space_bound}
In this section, we bound the space consumed by the sketch.
The main part is analyzing the space required for each compactor
as the number of compactors can be bounded by $\O(\log \varepsilon N)$ in a straightforward way (see \Cref{num_compactors}).
The relative compactors from~\cite{ReqSketch} were designed such that each of them fits into space of $\O(\varepsilon^{-1}\cdot \sqrt{\log \varepsilon N})$,
while the hard part is proving the right properties regarding the error, particularly the bound on the number of important compactions that involves convoluted charging (see Lemma~6.4 in~\cite{ReqSketch} which is an analogy of our \Cref{mbound}).

The adaptive compactors are, in a sense, opposite: They minimize the space requirements while maintaining several invariants, outlined in~\Cref{size}, which make bounding the number of important compactions straightforward (\Cref{coins} and \Cref{mbound}).
However, the space bound is not obvious on first sight. With stream updates only, it is quite intuitive that at the asymptotic behaviour is the same as for relative compactors, but proving the bound under full mergeability is more challenging.

The cornerstone of the analysis is \Cref{Kbound}, which gives a lower bound on the section length $K$. With this lemma holding, the rest of the analysis is straightforward. However, the lemma itself is the most challenging part of our analysis, requiring a careful argument using a suitable potential function. The proof builds upon the intuition that before doing a compaction of size $XK$, we must do $\approx 2^{X}$ compactions of smaller size. However, we need to deal with the fact that the parameters $K$ and $C$ can change after a merge. To formalize our intuition, we define a potential $\Phi$ that bounds the amount of compactions that happened on a given compactor since the last change of parameters.

For simplicity of notation, we sometimes write $\lapprox, \gapprox, \approx$ instead of $\O, \Omega, \Theta$, respectively.

\begin{lemma}[Lower bound on $K$]
\label{Kbound}
Consider an adaptive compactor at any level.
Let $P$ be the number of compactions performed on the compactor. For $P\ge 2$, we always have that
$$
K \in \Omega\left(\frac{C}{\log P}\right),
$$
while $K \in \Omega(C)$ for $P\le 1$.
\end{lemma}

\begin{proof}
We prove the statement by defining a suitable potential function $\Phi$ that for given state of a compactor returns a nonnegative real number. Then we prove that the potential of a compactor has following five properties.

\subsubsection*{Five properties of $\Phi$}
\begin{enumerate}[label=(P\arabic*)]
\item Adding new items to $B$ does not increase $\Phi$
\item Immediately before changing the values of $K$ and $C$ we have $\Phi \in 2^{\Omega\left(\frac{C}{K}\right)}$
\item Immediately after changing the values of $K$ and $C$ we have $\Phi = 0$
\item The compaction operation increases $\Phi$ at most by an additive constant
\item Merging compactors $\mathcal{C}_a, \mathcal{C}_b$ with potentials $\Phi_a, \Phi_b$ results in potential at most $\max(\Phi_a, \Phi_b)$
\end{enumerate}
These properties together imply that between changing $K$ to $K'$, $2^{\Omega\left(\frac{C}{K}\right)}$ compactions happen. As the number of compactions between changing from $K$ to $K'$ is at most $P$ and $K = 2 K'$, we have always $P \in 2^{\Omega\left(\frac{C}{K}\right)}$ and taking logarithm of both sides gives us the desired statement.
Note that the lemma holds even before any change of $K$ happens, as we initially set $K \approx C$. This also covers the case $P \leq 1$.

\subsubsection*{Definition of $\Phi$}

For a section $S[i]$ of $B$ let us denote by $\ma{i}$ that it is marked (in the canonical marking) and by $\un{i}$ that it is unmarked.
For any full section $S[i]$ we define its auxiliary potential $\varphi$ as follows (to recall naming of parts of $B$ see \Cref{parts}):
\begin{itemize}
	\item if $i$ lies in the left part, then
	$$
	\mathllap{\varphi(\ma{i}) = 2^{i/2} \qquad \text{and}}
	\mathrlap{\qquad \varphi(\un{i}) = - \sqrt{2} \cdot 2^{i/2},}
	$$
	\item if $i$ lies in the right part but not on the head, then
	$$
	\mathllap{\varphi(\ma{i}) = 0\qquad \text{and}}
	\mathrlap{\qquad \varphi(\un{i}) =  -(1+\sqrt{2})\cdot 2^{\frac{C}{4K}},}
	$$
	\item otherwise, $i$ lies on the head and we define
	$$
	\mathllap{\varphi(\ma{i}) = (1+\sqrt{2})\cdot 2^{\frac{C}{4K}}\qquad \text{and}}
	\mathrlap{\qquad \varphi(\un{i}) = 0.}
	$$
\end{itemize}
For section that is not full, the auxiliary potential $\varphi$ is always zero.
Now we define $\Phi$ for each prefix of sections of $B$. The potential of a prefix with no full section is 0 and for prefix $S[-\infty, i]$ where $i$ is a full section we have
$$
\Phi(i) = \max(0,\, \Phi(i-1) + \varphi(i)).
$$
The final potential of $B$ is naturally $\Phi(C/K-1) = \Phi$.

\subsubsection*{Key observations}

\begin{enumerate}[label=(O\arabic*)]
\item $\Phi$ is always nonnegative.

\item \label{geom} Let $\Delta \varphi(i) = \varphi(\ma{i}) - \varphi(\un{i})$. Notice that $\Delta \varphi$ forms a geometric series with quotient $\sqrt{2}$ in the left part and remains constant in the right part. More precisely, for each section $S[i]$ in the left part, we have
$$
\sqrt{2} \cdot \varphi(\ma{i}) = \varphi(\ma{i+1});\;\;\;
\sqrt{2} \cdot \varphi(\un{i}) = \varphi(\un{i+1});\;\;\;
\sqrt{2} \cdot \Delta \varphi(i) = \Delta \varphi(i+1)$$
and for all sections $i$ in the right part, we have
$$
\Delta \varphi(i) = \Delta \varphi(i+1).
$$

\item \label{swap} For $i < j$ suppose that section $S[i]$ is marked while $S[j]$ is unmarked. As $\Delta \varphi(i)$ is nondecreasing, if we unmark $S[i]$ and mark $S[j]$, the potential $\Phi$ does not decrease. Consequently, if we mark $S[i]$ and unmark $S[j]$, $\Phi$ does not increase.

\item \label{tail} We always have $\Phi(i) \leq (2 + \sqrt{2}) \cdot 2^{i/2}$, particularly the potential of the tail is at most $2 + \sqrt{2}$.

\end{enumerate}

\noindent
It remains to prove all the five properties (P1-P5) of $\Phi$.

\subsubsection*{(P1) Adding new items to $B$ does not increase $\Phi$.}
After adding an item to $B$, each marker either marks the same sections as before, or it moves one section to the left. Thus by observation \ref{swap}, $\Phi$ does not increase.

\subsubsection*{(P2) Immediately before changing the values of $K$ and $C$ we have $\Phi \in 2^{\Omega\left(\frac{C}{K}\right)}$.}
We change $K$ only when the number of unmarked sections equals one. Thus, one of the sections in the head is marked, and by the definition of the potential we have $\Phi \geq 2^{\frac{C}{4K}} \in 2^{\Omega\left(\frac{C}{K}\right)}$.

\subsubsection*{(P3) Immediately after changing the values of $K$ and $C$ we have $\Phi = 0$.}
By invariant \ref{whitehead}, all the items in the head are unmarked, thus contribute 0 potential. The items in the rest of the right part never contribute a positive potential. The left part is empty as it was just compacted by a special compaction. Finally, increasing $C$ preserves these properties.

\subsubsection*{(P4) The compaction operation increases $\Phi$ at most by an additive constant.}
For $K = 1$ this is trivial since there is no new marker, so we assume that $K > 1$.

Let $S[i]$ be the first section not removed by the compaction and let $\Phi_0$ be the potential before the compaction. The potential of the compacted prefix $S(-\infty, i)$ was originally some number $\Phi_C$ and after the compaction it is surely 0. Let the potential after the compaction be $\Phi_0 - \Phi_C + \Phi_m$ where $\Phi_m$ is the potential increase over the not compacted part caused by addition of the new marker. Our goal is to bound the overall potential increase $\Phi_m - \Phi_C$.

If the compaction is standard, then the newly created marker marks section $S[i]$ and all the other remaining sections are marked exactly as before the compaction (while items from sections $S[-\infty, i)$ are removed from $B$ during the compaction). Thus, we have $\Phi_m = \Delta \varphi(i) = (1 + \sqrt{2}) \cdot 2^{i/2}$. If the compaction is special, the number of marked sections in the right part increase by one and thus again we have $\Phi_m = \Delta \varphi(i) = (1 + \sqrt{2}) \cdot 2^{i/2}$.

It remains to bound $\Phi_C$ from below. Recall, that by invariant \ref{contiguous_sections}, all the compacted items not lying on the tail are marked. By observation \ref{geom}, the marked sections form a geometric series with quotient $\sqrt{2}$ and sum $(1 + \sqrt{2})\cdot 2^{i/2}$, thus if the number of marked sections was infinite, we had exactly $(1 + \sqrt{2})\cdot 2^{i/2}$ of potential from them.
While we have guaranteed marked sections only in the left part without the tail, the potential of the tail is at most $2 + \sqrt{2}$ by observation \ref{tail}.
Hence, we have $\Phi_C \geq (1 + \sqrt{2}) \cdot 2^{i/2} - (2 + \sqrt{2})$ and the overall potential increase is at most $2 + \sqrt{2}$.

\subsubsection*{(P5) \boldmath{}Merging compactors $\mathcal{C}_a, \mathcal{C}_b$ with potentials $\Phi_a, \Phi_b$ results in potential at most $\max(\Phi_a, \Phi_b)$}

Let us recall that for two compactors $\mathcal{C}_1 (B_1, C_1, K_1, M_1), \mathcal{C}_2 (B_2, C_2, K_2, M_2)$ with $K_1 \leq K_2$, their merge is defined as $B = B_1 \cup B_2$, $C = C_1$, $K = K_1$, and $M = M_1 \cup M_2$. If $K_1 = K_2$, let us name the compactors such that the smallest marker in $M_1$ is smaller than the smallest marker in $M_2$
Let us denote the potentials of the compactors $\mathcal{C}_1, \mathcal{C}_2$ by $\Phi_1$ and $\Phi_2$ respectively and let us denote the potential of the new compactor by $\Phi_{\text{new}}$. We prove that $\Phi_{\text{new}} \leq \Phi_1$.

Let us temporarily remove the smallest $2K_2$ items from $B_2$ before merging the compactors. By \Cref{markersRange}, the canonical marking of $B$ exists (as the removed items were unmarked in $\mathcal{C}_2$) and for each marker $m' \in B_1$ we have $\rank(m', B) \geq \rank(m', B_1)$. Let us add the removed $2K_2$ items to $B$ and let $m$ be the smallest of the markers from $B_2$. After the addition of items, $\rank(m, B)$ increases by $2K_2$. Let $S[i]$ be the last (rightmost) section marked by $m$ after the addition.

If there are some markers originally from $B_1$ smaller than $m$, rank of some of them may also increase by the addition. For the sake of analysis, let us remark the sections marked by those markers such that the sections are marked as before the addition. By observation \ref{swap}, this can only increase $\Phi$. After the remarking, the sections $S[i+1, i+2]$ are unmarked (as $K_2 \geq K$) and we still have $\rank(m', B) \geq \rank(m', B_1)$ for each marker $m'$ from $B_1$.

Now it suffices to prove, that $\Phi(i+2) = 0$. Then we have $\Phi_{\text{new}} \leq \Phi_1$, as $K = K_1$, the markers from $M_1$ have the same or larger rank in $B$ and the markers from $M_2$ contribute zero potential.

First note, that sections $S[i+1, i+2]$ do not lay on the head. If $K_2 > K_1$ (and thus $K_2 \geq 2K_1$), this is because $\rank(m, B)$ increased by at least $4K$ after the addition and the canonical marking existed before the addition. If $K_1 = K_2 = K$, this is because in $M$ there exists some marker $m'$ smaller then $m$ that marks some sections to the right of $S[i+2]$ and there are at least 2 unmarked sections to the right of the sections marked by $m'$ as $\rank(m', B) \geq \rank(m', B_1)$.

If both sections $S[i+1, i+2]$ lay in the right part (but not on the head), we have
$$
\varphi(i+1) + \varphi(i+2) = -(2 + 2\sqrt{2})\cdot 2^{\frac{C}{4K}}\;\;
\text{and}\;\;
\Phi(i) \leq (2 + \sqrt{2})\cdot 2^{\frac{C}{4K}-\frac{1}{2}}
$$
by observation \ref{tail}. This gives us $\Phi(i+2)=0$.

If we have $i = C/2K-2$ (thus $S[i+1]$ is the last section of the left part), we get
$$
\varphi(i+1) + \varphi(i+2) = -(2 + \sqrt{2}) \cdot 2^{\frac{C}{4K}}\;\;
\text{and}\;\;
\Phi(i) \leq (2 + \sqrt{2})\cdot 2^{\frac{C}{4K}-1}
$$
and finally if both sections $S[i+1, i+2]$ lay in the left part, we have
$$
\varphi(i+1) + \varphi(i+2) = -(2 + 2\sqrt{2}) \cdot 2^{\frac{i}{2}}\;\;
\text{and}\;\;
\Phi(i) \leq (2 + \sqrt{2})\cdot 2^{\frac{i}{2}}.
$$
This concludes the proof.
\end{proof}

\Cref{Kbound} together with invariant \ref{removeK}, stating that each compaction removes at least $K$ items from $B$, implies
the following bound on $P$, the number of compactions performed by the compactor:
$$
P \leq \frac{N}{K} \lapprox \frac{N \log P}{C} \lapprox \varepsilon N \log P\,,
$$
which gives:

\begin{obs}[Bound on $P$]
\label{Pbound}
$
\log P \in \O\left(\log(\varepsilon N)\right).
$
\end{obs}

We now bound the number of compactors $H$, which is also the number of levels.
Recall that a new compactor is created when the highest compactor performs its first compaction and compactor never performs compaction before it is full.

Since for each $x$ items moved from level $h$ to level $h+1$, there are exactly $x$ items deleted from level $h$,
the input stream $I_h$ of the level-$h$ compactor has at most half the number of items inserted to level $h-1$.
By induction, we obtain $\abs{I_h} \leq N/2^h$.
Using this together with $C \geq \varepsilon^{-1}$, we get:

\begin{obs}[Number of compactors]
\label{num_compactors}
$
H \leq \log(\varepsilon N)+2\,,
$
\end{obs}

\begin{lemma}[The space bound]\label{lem:spaceBound}
The space bound on the whole sketch after processing $N$ items is
$$
\spa \in \O\left(\min\left\{\;\log^{1.5}(\varepsilon N)\cdot \left(\varepsilon^{-1}\sqrt{\ln \delta^{-1}} + \ln \delta^{-1}\right) \;,\;
\log(\varepsilon N) \cdot \varepsilon^{-2} \ln \delta^{-1} \;\right\}\right).
$$
\end{lemma}

\begin{proof}
By invariant \ref{Cinv}, we have $CK \approx \varepsilon^{-2}\ln \delta^{-1} + \ln^2 \delta^{-1}$. This together with \Cref{Kbound} and \Cref{Pbound} gives a bound on $C$:
\begin{align*}
C &\lapprox \frac{\varepsilon^{-2}\ln \delta^{-1} + \ln^2 \delta^{-1}}{K}\\
C &\lapprox \frac{(\varepsilon^{-2}\ln \delta^{-1} + \ln^2 \delta^{-1}) \cdot \log P }{C}\\
C &\lapprox \sqrt{\log P} \cdot (\varepsilon^{-1} \sqrt{\ln \delta^{-1}} + \ln \delta^{-1})\\
C &\lapprox \sqrt{\log(\varepsilon N)} \cdot (\varepsilon^{-1} \sqrt{\ln \delta^{-1}} + \ln \delta^{-1}).
\end{align*}

At the same time we have $C \lapprox \varepsilon^{-2} \ln \delta^{-1}$ as $K \geq 1$.
Surely $C$ bounds the space occupied by each compactor and as by \Cref{num_compactors} the number of compactors is $\O(\log(\varepsilon N))$, the statement of the lemma follows.
\end{proof}

\subsection{The error bound}
\label{error_bound}
In this section, we prove the error guarantee of the sketch. The analysis is adapted from~\cite{ReqSketch}, but is largely simplified. It consist of three lemmas: \Cref{tech} (analogy to Lemma 6.7 in \cite{ReqSketch}) states, that for arbitrary item $y$, its rank decreases exponentially as we go up the levels. In \Cref{init} (analogy to Lemma 6.11 in \cite{ReqSketch}) we prove a weaker error guarantee of $\pm \rank(y)$ and in \Cref{final_error_estimate} we use the weaker guarantee to obtain the final bound $\pm \varepsilon\rank(y)$.
However, we avoid the most complicated parts of the proof in Section~6 of~\cite{ReqSketch}, namely,
 Lemma~6.12 in~\cite{ReqSketch}, and require less definitions, e.g., we do not need the quantities $q^i_h$ and $z^i_h$.

The following fact is a simple corollary of Hoeffding bound.
\begin{fact}[Corollary of Hoeffding bound]
\label{fact}
Let $X_1 \dots X_n$ be independent random variables such that each $X_i$ attains values $a_i, -a_i$ with equal probability and let $X = \sum_i X_i$. Then for all $t>0$ we have
$$
\pr \bigl[X \geq t\bigr] \leq \exp \left( -\frac{2t^2}{\var[X]} \right)
$$
and
$$
\pr \bigl[\abs{X} \geq t\bigr] \leq 2\exp \left( -\frac{2t^2}{\var[X]} \right).
$$
\end{fact}

For the whole error analysis, let us fix an arbitrary item $y$ and let us recall that the error of an item $y$ is defined as $\err(y) \eqdef \estrank(y) - \rank(y)$ where $\rank(y)$ is the true rank of $y$ in the input stream and $\estrank(y) \eqdef \sum_{h=0}^{H-1}2^h \rank(y, B_h)$. Thus, after each compaction on level $h$, the error either increases by $2^h$ or decreases by $2^h$ or stays the same.

\begin{defn}[Important items and compactions]
Let us call all items $z \leq y$ \emph{important items}. We say that compaction is important if it affects $\err(y)$. Let us denote the number of important compactions on level $h$ by $m_h$.
\end{defn}

If the rank of $y$ among the compacted items is even, the compaction is not important, as regardless the random choice, half of the important compacted items get removed and half is promoted to the next level. On the other hand, if the rank of $y$ among the compacted items (on level $h$) is odd, the error either increases or decreases by $2^h$.

\begin{obs}[Important compactions]
\label{imp}
Compaction is important if and only if the rank of $y$ among the compacted items
is odd. 
\end{obs}

The following is a trivial bound on the number of important compaction on given level. We later devise tighter bound in \Cref{mbound}.

\begin{obs}[Bound on $m_h$]
\label{m_simple_bound}
The previous observation implies that each important compaction removes at least
one important item from $B$. This gives us for each $h$
$$
m_h \leq \rank(y, I_h)\,,
$$
where $I_h$ is the input stream of level $h$.
\end{obs}

\Cref{tech} claims, that for any item $y$ its rank decreases exponentially as we go up the levels. Although this is intuitively true, it is not trivial to prove that it holds with high probability and in fact this is the hardest part of the proof of the error bound.

For the proof of the lemma, we need a definition of the \emph{critical level}. This is intuitively the highest level such that any important item ever reaches it, conditioning on that the $y$'s rank indeed decreases exponentially.

\begin{defn}(Critical level)
\label{critical}
Let $H_y$ be the lowest level $h$ such that 
$$2^{-h+1}\rank(y) \leq 2^4 \ln \delta^{-1}.$$
\end{defn}

For the proof we also need following corollary of the definition of critical level. Informally, it says that the next-to critical level must contain enough important items (again conditioning on the statement of the lemma).

\begin{obs}
\label{level}
From the minimality of $H_y$, if $H_y > 0$, we have $2^{-(H_y-1)+1}\rank(y) > 2^4 \ln\delta^{-1}$, which reduces to
$$
2^{-H_y - 2}\rank(y) > \ln\delta^{-1}.
$$
\end{obs}

\begin{lemma}[Ranks decrease exponentially]
\label{tech}
With probability at least $1-\delta/4$ we have simultaneously for all levels $h$
$$\rank(y, I_h) \leq 2^{-h+1}\rank(y)$$
and for all levels $\ell >  H_{\ell}$
$$
\rank(y, I_{\ell}) = 0.
$$
\end{lemma}

\begin{proof}
The proof is by induction over $h$, with the base case $h = 0$ holding as $\rank(y, I_0) = \rank(y)$.
For the rest of the proof, consider $h > 0$.
First observe, that it suffices to prove the lemma for $h \leq H_y$. This is true because conditioning on the statement holding for all $h \leq H_y$, by \Cref{critical} we have $\rank(y, I_{H_y}) \leq 2^4 \ln \delta^{-1}$ and by invariant \ref{Clog} we always have $2^4 \ln \delta^{-1} \leq C/2$. By invariant \ref{protected} the smallest $C/2$ items are never compacted, thus $\rank(y, I_{H_y+1}) = 0$.

We prove for all $h \leq H_y$, that if for all $\ell < h$ we have 
\begin{equation}
\label{cond}
\rank(y, I_\ell) \leq 2^{-\ell+1}\rank(y),
\end{equation}
then with probability at least $1 - \delta \cdot 2^{h-H_y-3}$ we have 
\begin{equation}
\label{stat}
\rank(y, I_h) \leq 2^{-h+1}\rank(y). 
\end{equation}
This is sufficient to prove the lemma as by the union bound, \cref{stat} is simultaneously true for all $h$ with probability at least $1 - \delta \sum_{h=0}^{H_y}2^{h-H_y-3} > 1 - \delta/4$. 

Let us investigate the dependence of $\rank(y, I_\ell)$ on $\rank(y, I_{\ell -1})$.
Letting $ O_{\ell-1}$ be the output stream of the compactor at level $\ell - 1$,
it holds that $\rank(y, I_\ell) = \rank(y, O_{\ell-1})$ so we are comparing the input and the output of the compactor on level $\ell - 1$. Some of the items from $I_{\ell - 1}$ can stay in the buffer $B_{\ell-1}$, the rest are subject to compactions. Recall that an important compaction is a compaction that affects $\err(y)$ and by \Cref{imp} these are exactly those compactions where the rank of $y$ among the compacted items is odd. This means that any non-important compaction always sends half of the important items to the output, and any important compaction always sends one important item more or less than half with equal probability. Let us recall that the number of important compactions on level $\ell$ is $m_{\ell}$. Neglecting the items staying in the buffer, we can bound the rank of $y$ in the input stream $I_\ell$ as
$$
\rank(y, I_{\ell}) = \rank(y, O_{\ell-1}) \leq \frac{1}{2}(\rank(y, I_{\ell -1}) + \bin(m_{\ell -1})),
$$
where $\bin(m)$ is a sum of $m$ independent random variables taking values from $\{-1, 1\}$ with equal probability,
generated with the same random bits as used in the level-$(\ell-1)$ compactor for selecting even/odd-indexed items.

Let us use this bound recursively for all $\rank(y, I_{\ell})$ where $0 \leq \ell \leq h$. Let $Y_0 \eqdef \rank(y)$ and for $0 < \ell \leq h$ let 
$$
Y_{\ell} \eqdef \frac{1}{2}(Y_{\ell - 1} + \bin(m_{\ell - 1})).
$$
It follows that $\rank(y, I_h) \leq Y_h$. Thus, it suffices to prove that with probability at least $1 - 2^{h-H_y-3}\delta$ we have
$$
Y_h \leq 2^{-h+1}\rank(y).
$$
By unrolling the definition of $Y_h$ we obtain
$$
Y_h = 2^{-h}\rank(y) + \sum_{\ell=0}^{h-1}2^{-h+\ell}\bin(m_{\ell}).
$$
Note that the first summand has a fixed value and let us denote the second summand by $Z_h \eqdef \sum_{\ell=0}^{h-1}2^{-h+\ell}\bin(m_{\ell})$. It suffices to prove that 
\begin{equation}
\label{z}
\pr[Z_h  > 2^{-h}\rank(y)] \leq 2^{h-H_y-3}\delta.
\end{equation}
Note that the expression $Z_h$ can be viewed as a sum of random variables each attaining values $\pm a_i$ for some constant $a_i$ with equal probabilities, thus we can use \Cref{fact} on $Z_h$ to prove \Cref{z}. 

To do this, we need to bound $\var[Z_h]$. As each important compaction removes at least one important item from the buffer, we have for each $\ell < h$: $m_\ell \leq \rank(y, I_\ell)$ (\Cref{m_simple_bound}) and by \cref{cond} we have $\rank(y, I_\ell) \leq 2^{-\ell+1}\rank(y)$. We also have $\var[\bin(n)]=n$. All this together gives us following bound:
\begin{align*}
\var[Z_h]
\leq \sum_{\ell = 0}^{h-1} 2^{-2h+2\ell}m_{\ell}
\leq \sum_{\ell = 0}^{h-1} 2^{-2h+2\ell}\rank(y, I_\ell)
&\leq \sum_{\ell = 0}^{h-1} 2^{-2h+2\ell}2^{-\ell+1}\rank(y)\\
&= \sum_{\ell = 0}^{h-1} 2^{-2h+\ell+1}\rank(y) \\
&\leq 2^{-h+1}\rank(y)
\end{align*}

Now we apply \Cref{fact} with $t = 2^{-h}\rank(y)$ together with \Cref{level} ($2^{-H_y - 2}\rank(y) > \ln\delta^{-1}$) to prove \Cref{z}:
\begin{align*}
\pr[Z_h  > 2^{-h}\rank(y)]
< \exp \left( -\frac{2^{-2h+1}\rank^2(y)}{\var[Z_h]} \right)
&\leq \exp \left( -\frac{2^{-2h+1}\rank^2(y)}{2^{-h+1}\rank(y)} \right) \\
&= \exp \left( -2^{-h}\rank(y) \right) \\
&= \exp \left( -2^{-h + H_y + 2}\cdot 2^{-H_y - 2}\rank(y) \right) \\
& \leq \exp \left( -2^{-h + H_y + 2} \ln \delta^{-1} \right)\\
&= \delta^{2^{-h + H_y + 2}}
\leq  2^{h-H_y-3}\delta
\end{align*}
The last inequality follows from $\delta \leq 1/2$. This concludes the proof.

\end{proof}

By a quite straightforward application of \Cref{tech} and \Cref{fact} we obtain a crude error bound, which is however important for proving \Cref{remaining} which is in turn an essential tool in proving the tight error bound (\Cref{final_error_estimate}).

\begin{lemma}[Initial error bound]
\label{init}
Conditioning on the bound from \Cref{tech} holding, with probability at least $1 - \delta/2$ we have
$$\err(y) \leq \rank(y).$$
\end{lemma}

\begin{proof}
From the definition of important compaction, we have $\err(y) = \sum_{h=0}^{H-1} 2^h \bin(m_h)$
and by \Cref{tech}, there are no important compactions above level $H_y$ and so we have
$$
\err(y) = \sum_{h=0}^{H_y} 2^h \bin(m_h)
$$
as we are conditioning on the bound from \Cref{tech}.

As in the proof of \Cref{tech}, we can see that $\err(y)$ meets the conditions of \Cref{fact} and so we can use it to prove the statement. Thus, we need to bound $\var[\err(y)]$. For this purpose we use \Cref{m_simple_bound} ($m_h \leq \rank(y, I_h)$), \Cref{tech} ($\rank(y, I_h) \leq 2^{-h+1}\rank(y)$) and then \Cref{level} ($2^{-H_y-2}\rank(y)>\ln \delta^{-1}$):
\begin{align*}
\var[\err(y)]
= \sum_{h=0}^{H_y} 2^{2h} m_l 
\leq \sum_{h=0}^{H_y} 2^{2h} \rank(y, I_h) 
\leq \sum_{h=0}^{H_y} 2^{2h}2^{-h+1}\rank(y) 
&= \sum_{h=0}^{H_y} 2^{h+1}\rank(y) \\
& \leq 2^{H_y + 2}\rank(y) \\
&= 2^{H_y + 2} \rank^{-1}(y)\cdot \rank^2(y) \\
&\leq \rank^2(y)/\ln \delta^{-1}
\end{align*}

Finally we apply \Cref{fact} with $t = \rank(y)$:
$$
\pr \bigl[\err(y) \geq \rank(y)\bigr] \leq
\exp \left( -\frac{2\rank^2(y)}{\var[\err(y)]} \right) \leq 
\exp \left( -2 \ln \delta^{-1} \right) = 
\delta^2 \leq \delta/2
$$
The last inequality follows from $\delta \leq 1/2$. This concludes the proof.

\end{proof}

\begin{obs}
\label{remaining}
Let $E_h$ be the number of important items remaining in $B_h$ of the final sketch. Conditioning on the bound from \Cref{init} we have
$$
\sum_h 2^h E_h = \estrank(y) \leq 2 \rank(y).
$$
\end{obs}

We use \Cref{coins} to devise a better bound on the number of important compactions. The idea is straightforward -- as we remove $K$ important items for each important compactions, we have $m_h \leq \rank(y, I_h) / K$. However, the section length $K$ decreases over time and using the value of $K$ of the final sketch is insufficient for the proof of \Cref{final_error_estimate}. Thus, we use the value of $K$ from the ``last'' important compaction on level $h$.

\begin{defn}
Let $K_{\comp}$ and $C_{\comp}$ be respectively the values of $K$ and $C$ during a compaction $\comp$ and let $\comp_{y, h}$ be a level-$h$ important compaction $\comp$ such that $K_{\comp}$ is minimal (and consequently $C_{\comp}$ is maximal).
For simplicity, let us denote $K_{y, h} \eqdef K_{\comp_{y, h}}$ and $C_{y, h} \eqdef C_{\comp_{y, h}}$.
\end{defn}

The following bound on the number of important compactions is an analogy of Lemma 6.4 in \cite{ReqSketch}, which is proven by an involved charging argument. With adaptive compactors, this is just a simple corollary of (also simple) \Cref{coins}.

\begin{obs}[Better bound on $m_h$]
\label{mbound}
By \Cref{coins}, we have that
$$
m_h \leq \frac{\rank(y, I_h)}{K_{y, h}}.
$$
\end{obs}

The following observation is the reason why we can not simply use the value of $K$ of the final sketch in \Cref{mbound}. As $C$ grows over time, \Cref{restInC} is not necessarily true for the final value of $C$, but it is true for $C_{y, h}$. This allows as to bound the number of important compactions by the number of important items remaining in the buffer of the final sketch and use \Cref{remaining} to bound the error.

\begin{obs} \label{restInC}
We have that $C_{y, h}/2 \leq E_h$.
\end{obs}
\begin{proof}
After an important compaction $\comp$ there are at least $C_{\comp}/2$ important items in $B$ and as we never remove the smallest $C/2$ items by invariant \ref{protected} and $C$ is nondecreasing in time by invariant \ref{monotonicity}, there must be at least $C_{\comp}/2$ important items present in the final sketch.
\end{proof}

\begin{lemma}[Final error estimate]
\label{final_error_estimate}
With probability $1 - \delta$ we have
$$\abs{\err(y)} \leq \varepsilon\rank(y).$$
\end{lemma}

\begin{proof}
We prove that if both bounds from \Cref{tech,init} hold, we have $$\abs{\err(y)} \leq \varepsilon\rank(y)$$ with probability at least $1-\delta/4$, and the statement then follows from the union bound. Analogously to previous proofs, we first bound $\var[\err(y)]$ and then we apply \Cref{fact} to prove the desired statement.

Using the previously obtained bounds, we start by \Cref{mbound} (which conditions on \Cref{init}), the second inequality (\ref{two}) is by \Cref{tech}, the third (\ref{three}) by invariant \ref{Cinv}, the next inequality (\ref{four}) is by \Cref{restInC} and the last one (\ref{five}) is by \Cref{remaining}.

\begin{align}
\var[\err(y)] &= 
\sum_{h} 2^{2h}m_h \\
&\leq \sum_{h} 2^{2h} \frac{\rank(y, I_h)}{K_{y, h}} \label{one}\\
&\leq \sum_{h} 2^{2h} \frac{\rank(y)2^{-h+1}}{K_{y, h}}
= \sum_{h} 2^{h+1} \frac{\rank(y)}{K_{y, h}} \label{two}\\
&\leq \sum_{h} 2^{h+1} \frac{\varepsilon^2 \rank(y) C_{y, h}}{2^3 \ln \delta^{-1}} =
\frac{\varepsilon^2 \rank(y)}{2\ln \delta^{-1}} \sum_{h} 2^h C_{y, h}/2 \label{three}\\
&\leq \frac{\varepsilon^2 \rank(y)}{2 \ln \delta^{-1}} \sum_{h} 2^h E_h \label{four}\\
&\leq \frac{\varepsilon^2 \rank^2(y)}{\ln \delta^{-1}}\label{five}
\end{align}

Now it suffices to apply second part of \Cref{fact} with $t = \varepsilon\rank(y)$:
\begin{align*}
\pr \biggl[\abs{\err(y)} \geq \varepsilon \rank(y)\biggr] \leq
\,2 \exp \left( -\frac{2\varepsilon^2\rank^2(y)}{\var[\err(y)]} \right) \leq 
&2 \exp \left( -\frac{2 \ln \delta^{-1}\varepsilon^2\rank^2(y)}{\varepsilon^2 \rank^2(y)} \right) \\
&= 2\delta^{2} \leq \delta/4
\end{align*}
The last inequality follows from $\delta \leq 1/4$.

\end{proof}

\section{Space bound improvements in special cases}
\label{improvements}

In this section, we present minor improvements of the space bound over original ReqSketch~\cite{ReqSketch} in two special cases.
First, our sketch uses near-optimal space if the summary is produced by merging and the merge tree is balanced, even approximately.
Second, we provably achieve the optimal space bound for any reverse-sorted input.

\subsection{Better space bound with balanced merging}
As mentioned earlier, the update operation can be viewed as a merge with a trivial sketch representing one item. Thus, we can assume that the sketch is build only by merges. The history of the sketch can be visualized by a binary tree where the leaves are trivial one-item sketches, each internal node represents the state of the sketch after merging its two children, and the root represents the final state of the sketch. 
For a fixed level $h$, at most one compactions is performed during any merge operation. Thus for fixed $h$, the internal nodes also represent all the compactions that ever happen on level $h$ (i.e., each level-$h$ compaction is represented by a unique node).

In the proof of \Cref{Kbound} we have that the potential $\Phi$ after a merge is at most maximum of the potentials of the two merged compactors. This implies that it is sufficient to count in variable $P$ only the compactions performed in the subtree of one of the node's children in the merge tree. Formally, for any level $h$ we can redefine $P$ from the statement of \Cref{Kbound} such that for trivial one-item sketch $P = 0$ and for a sketch represented by internal node $i$ with children $l$ and $r$ we define $P$ as maximum of $P_r, P_l$ plus 1 if node $i$ represents level-$h$ compaction.

For this definition of $P$, \Cref{Kbound} still holds and for each level $h$, $P_h$ in the final sketch is bounded by the depth of the tree. Thus for an approximately balanced merge tree, we have $P \in \O(\log N)$ and so $\log P \in \log \log N$. This gives us (by the same calculation as in \Cref{lem:spaceBound}) space bound
$$
\spa \in \O\left(\log(\varepsilon N) \cdot \sqrt{\log \log N} \cdot
(\varepsilon^{-1}\sqrt{\ln \delta^{-1}} + \ln \delta^{-1}) \right)
$$
for the whole sketch, which is only by factor $\sqrt{\log \log N}$ from the optimum
(for constant $\delta$).

More generally, $P$ is bounded by the height $\mathcal{H}$ of the merge tree, which leads to a space bound of $\O\left(\log(\varepsilon N) \cdot \sqrt{\mathcal{H}} \cdot
(\varepsilon^{-1}\sqrt{\ln \delta^{-1}} + \ln \delta^{-1}) \right)$.

\subsection{Optimal space bound for reverse-sorted inputs}

We named our new compactors \emph{adaptive compactors}, because they perform better on ``nice'' inputs. 
Here, we demonstrate it on a simple example of any reverse-sorted input, i.e., a strictly decreasing permutation.

Let us first consider the behaviour of the level-0 compactor on the reverse-sorted input. When the compactor contains the largest $C$ items of the input stream, the first compaction removes the largest $K$ items from the compactor and marks the next $K$ largest ones. The items that come after the compaction are all smaller than all the present items, thus when the compactor is full again, the largest $K$ items are marked. The second compaction removes all the marked items, marks $K$ largest remaining items and we find ourselves in the same situation again. Thus, each compaction removes all the marked items and marks the $K$ largest remaining items, which remain the largest to the next compaction and are subsequently removed. Thus the level-0 compactor never performs a special compaction and never changes the value of $K$ and $C$, so its size remains $C_0 \approx \varepsilon^{-1}\sqrt{\ln \delta^{-1}} + \ln \delta^{-1}$.

Moreover, observe that for any items $y < z$, if $z$ comes before $y$ in the input stream, the order of these two items never changes, meaning that $z$ comes before $y$ in the input of any compactor (as long as they are both present in the sketch). This is simply because each compaction always removes the largest items from $B$. It particularly means that if the input stream is reverse-sorted, then input streams of all the compactors are reverse-sorted and all the compactors keep their initial constant size. Thus the space bound becomes 
$$
\spa \in \O\left(\log(\varepsilon N) \cdot (\varepsilon^{-1} \sqrt{\ln \delta^{-1}} + \ln \delta^{-1}) \right)
$$
which is optimal with respect to $\varepsilon$ and $N$.

A similar analysis can be done for any stream that contains a lot of pairs $y, z$ such that $y < z$ and $z$ comes before $y$ in the input stream. Whenever a new item which is smaller than some marks comes to a compactor, all these marks ``shift to the left'' in the optimal marking. Thus the closer is the stream to the reverse-sorted case, the closer is the space of the sketch to the optimum.

\section{Conclusions}

We have proposed adaptive compactors as a new building block for mergeable relative-error quantile sketches,
replacing relative compactors from~\cite{ReqSketch}, while retaining their main properties such as accuracy, memory footprint, and update time.
The main point is to get more accessible and intuitive proofs of the space bound of $\O(\varepsilon^{-1}\cdot \log^{1.5} \varepsilon n)$
(with constant probability $\delta$ of a too large error),
even when the sketch is created by an arbitrary sequence of pairwise merge operations. 
That is, we avoid the most involved calculations required for full mergeability  in~\cite{ReqSketch},
and only need to analyze the space consumed by adaptive compactors using a new intuitive potential function.
The upside of our approach is that it allows to flexibly analyze the sketch with respect to special inputs without too much technical work,
as we demonstrate in \Cref{improvements}.

The main open problem is to improve the space bound towards the lower bound of $\Omega(\varepsilon^{-1}\cdot \log \varepsilon n)$~\cite{ReqSketch}
in the general mergeability setting.
We believe our paper makes a step in this direction by making the analysis for full mergeability more accessible and flexible.
More specifically, we ask if our techniques can be combined
with the approach in~\cite{GribelyukSWY25}, which gives a near-optimal bound of $\widetilde{O}(\varepsilon^{-1}\cdot \log n)$
in the streaming setting, but does not deal with merging sketches.
One of the main challenges of using adaptive compactors instead of relative compactors in the algorithm of~\cite{GribelyukSWY25}
is that adaptive compactors do not straightforwardly support the reset operation that shrinks the compactor to the initial state.

\label{discussion}

\printbibliography
 





\end{document}